\documentclass[12pt, onecolumn, draftclsnofoot, peerreview]{IEEEtran}
\usepackage{dsfont} 
\usepackage{graphicx}  
\usepackage{url}
\usepackage[cmex10]{amsmath}
\usepackage{psfrag,amssymb,amsfonts}

\newtheorem {definition}{Definition}

\newtheorem {lemma}{Lemma}
\newtheorem {theorem}{Theorem}
\newtheorem {corollary}{Corollary}

\newtheorem {proposition}{Proposition}

\begin{document}

\title{Analysis of a CSMA-Based Wireless Network: Feasible Throughput Region and Power Consumption}


\author{Fu-Te Hsu,~\IEEEmembership{Student Member,~IEEE,} and Hsuan-Jung Su,~\IEEEmembership{Member,~IEEE}
\thanks{The material in this paper was presented in part at IEEE International Symposium on Personal, Indoor and Mobile Radio Communications (PIMRC) 2009, Tokyo, Japan.}
\thanks{The authors are with the Department of Electrical Engineering
and Graduate Institute of Communication Engineering, National Taiwan University, Taipei, Taiwan, 10617
                  (email: f96942059@ntu.edu.tw, hjsu@cc.ee.ntu.edu.tw)}}

\maketitle

\begin{abstract}
We analytically study a carrier sense multiple access (CSMA)-based network.
In the network, the nodes have their own average throughput demands for transmission to a common base station. The CSMA is based on the request-to-send (RTS)/clear-to-send (CTS) handshake mechanism. Each node individually chooses its probability of transmitting an RTS packet, which specifies the length of its requested data transmission period. The RTS packets transmitted by different nodes in the same time slot
interfere with one another, and
compete to be received by the base station. If a node's RTS has the received
signal to interference plus noise ratio (SINR)
higher than the capture ratio, it will be successfully received. The node will then be granted the data transmission period. The transmission probabilities of RTS packets of all nodes will determine the average throughput and power consumption of each node. The set of all possible throughput demands of nodes that can be supported by the network is called the feasible throughput region. We characterize the feasible throughput region and provide an upper bound on the total power consumption for any throughput demands in the feasible throughput region. The upper bound
corresponds to
one of three points in the feasible throughput region depending on the fraction of time occupied by the RTS packets.


\end{abstract}


\section{Introduction}\label{Introduction}

The progress of wireless network technologies has provided ubiquitous services.
Medium access control (MAC) is one of the important mechanisms that contribute to the success of wireless networks.
For example, in wireless local area network (WLAN), the Carrier Sense Multiple Access with Collision Avoidance (CSMA/CA) and its variants are used as the MAC protocol (see more information in the IEEE 802.11 standards \cite{IEEE802.11}).%

Several studies analyzed the performance of IEEE 802.11-based networks, and some methods were proposed to improve the efficiency of channel utilization and power consumption \cite{Bianchi98}\cite{Bianchi00}\cite{p-persistent_802.11_TON00}\cite{Dynamic_802.11_JSAC00}\cite{Efficiency_EnergyConsumption_TMC02}.
In \cite{Bianchi00}, Bianchi developed a discrete-time Markov chain model to describe the evolution of the 802.11 backoff process. In \cite{p-persistent_802.11_TON00}, Cal\`{i}, Conti and Gregori used the $p$-persistet model instead of the standard binary exponential backoff in IEEE 802.11 due to its simplicity for analytical studies. The difference there is in the selection of the backoff interval which is sampled from a geometric distribution with parameter $p$. It was shown that the $p$-persistent IEEE 802.11 can closely approximate the standard protocol. In \cite{Capacity802.11_WirelessNetwork01}, Tay and Chua adopted a different modeling approach based on average values for analytical study. They derived closed-form approximations for the collision probability and maximum throughput.
All these works considered the saturated throughput of the system.

Later, the analytical model in \cite{p-persistent_802.11_TON00} was further explored.
In \cite{Dynamic_802.11_JSAC00}, a distributed tuned backoff was proposed and analyzed in depth.
In \cite{Efficiency_EnergyConsumption_TMC02}, the optimal $p$ value that maximizes the throughput and minimizes energy consumption was derived for a $p$-persistent CSMA model.

Recent advances in signal processing enabled the possibility of receiving a packet when multiple nodes are transmitting packets simultaneous. This is known as the multipacket reception (MPR) capability \cite{MPR_Mag01}. There have been some studies showing that the system performance of the earlier WLAN design based on the collision channel (i.e., packets collide when more than one node transmit) is not optimal and can be enhanced with MPR \cite{MPR_in_WLAN_ICC06}\cite{MPR_in_WLAN_TMC09}. The system performance of WLAN can also be improved by utilizing multi-user diversity. The readers are referred to \cite{MultiuserDiversity_in_WLAN} and the references therein for this issue.

From a completely different point of view, Lee {\it et al} \cite{Reverse_engineering_MAC07} used the tool of game theory \cite{GameTheory94}\cite{GameTheory91} and discovered
that the nodes in the network are participating implicitly in a noncooperative game with appropriate utility functions in backoff-based MAC protocols.
Game theory is a useful tool in analyzing distributed networks with self-configuring nodes and Nash equilibrium is a solution concept of a game. The readers are referred to \cite{GameTheory_and_the_Design} for the application of game theory in wireless networks.

Most studies on the performance analysis of a wireless network were under the condition of homogeneous nodes, i.e., the nodes use the same transmission probability, with the same transmission rate, etc. In this paper, we study a simple CSMA-based wireless network with heterogeneous nodes. In this multiple access network, each node has its own throughput demand to a common base station, and chooses its transmission probability individually to satisfy its throughput demand. The MAC protocol considered is CSMA with the request-to-send (RTS)/clear-to-send (CTS) handshake mechanism which was originally introduced to solve the hidden terminal problem (see the IEEE 802.11 standards \cite{IEEE802.11} for a survey).

The objective of our paper is similar to that of \cite{Efficiency_EnergyConsumption_TMC02}, that is, to analyze the throughput and power consumption of a CSMA-based network. However, there are some differences.
First, we consider heterogeneous nodes that node $i$ uses $p_i$-persistent CSMA in the network.
Second, the numbers of slots occupied by data transmissions were assumed to be identical and independent geometric random variables for all nodes in \cite{Efficiency_EnergyConsumption_TMC02}. In our paper, the lengths of data transmission periods differ from user to user. 
Third, we consider the CSMA model with the RTS/CTS handshake mechanism which is different from the CSMA model in \cite{Efficiency_EnergyConsumption_TMC02}.

Our paper is also similar to \cite{Menache_RTS/CTS} in the CSMA model. However,
we incorporate into the network one particular MPR model, the signal-to-interference-plus-noise-ratio (SINR) capture model, which is more general than the collision channel in \cite{Menache_RTS/CTS}. In addition, in \cite{Menache_RTS/CTS} the lengths of data transmission periods were the same for all nodes. Also, \cite{Menache_RTS/CTS} focused on
the analysis of convergence of the distributed algorithm which the authors significantly strengthened from their earlier work \cite{Aloha_JSAC08}, whereas our paper aims to characterize the the feasible region of throughput demands (which will be called the feasible throughput region in the remainder of this paper for brevity) and power consumption in the network.

Other related works are \cite{QosRegion_Boche04}\cite{QosRegion_Catrein04} which characterized the feasible quality of service (QoS) region, or SINR region, that can be supported by CDMA networks with power control.
However, those results can not be applied to our model in which the number of users interfering with one another is random.

\subsection{Contributions and Organization of this Paper}
This paper is devoted to analyzing the throughput and power consumption of a CSMA-based network.
With appropriate formulation through a game-theoretic approach, the traditional system (optimal) performance is contained in the Nash equilibrium.
We study the Nash equilibrium point, which, in the game model considered, is a vector of probabilities of transmitting requests (RTS packets) for heterogenous nodes in the CSMA network. The properties of the Nash equilibrium are then used to derive the feasible throughput region of the network. We further derive an upper bound on the total power consumption for any throughput demands in the feasible throughput region, and show that the upper bound
corresponds to one of three points in the feasible throughput region depending on the fraction of time occupied by the RTS packets.



The paper is organized as follows.
We describe the details of our network model in Section \ref{Model}, and formulate the problem in Section \ref{Problem Formulation}. The main results of this paper on the analysis of feasible throughput region and power consumption of the network are derived in Section \ref{Optimum Transmission Periods}. Finally, some conclusions are given in Section \ref{Conclusion}.

\section{The Network Model}\label{Model}
We consider a wireless network, where $n$ nodes transmit data to a common base station (BS) over a shared channel.
Time is slotted. Nodes intending to send data ask for the permission to transmit from the BS by sending an RTS packet. The BS responds with a CTS packet granting the use of the channel to at most one node at a time. Let the total duration of this two-way handshake be $T_0$ slots. If no node is granted the permission to send data, the two-way handshake is repeated for the next $T_0$ slots. If node $i$ is granted the permission, it can send its data without the interruption from the other nodes for a duration of $T_i$ slots, where $T_i$ is specified in the RTS packet sent by node $i$. The transmission power is $P_T$ for all nodes, for the RTS as well as the data packets. Without loss of generality, the transmission data rate is defined as one \textbf{data packet per slot}, where each data packet contains the same amount of data.

Independent \emph{Rayleigh} fading channels between nodes and the BS are assumed. The period of exchanging RTS and CTS is called the \emph{handshake phase}, and the period of data transmission is called the \emph{transmission phase}. Fig.~\ref{RTS_CTS} illustrates the CSMA with the RTS/CTS handshake mechanism.
We also assume that all nodes always have data to send as in \cite{Bianchi00}\cite{Efficiency_EnergyConsumption_TMC02}\cite{p-persistent_802.11_TON00} to analyze the network.

\textbf{Reception model}:

In each handshake phase, the BS can successfully receive the RTS packet with SINR larger than the \emph{capture ratio} $b$, and grant the permission to the corresponding node. We assume that $b>1$ (which is common for most systems except the spread spectrum systems), so at most one node is granted the permission.

\textbf{Behavior of Nodes}:

We associate node $i$ with an average throughput demand $\rho_i$ (in terms of the average number of successfully received data packets per slot), and assume that node $i$ chooses a request probability $p_i$ such that it randomly transmits an RTS packet with probability $p_i$ in every handshake phase.
%
The request probability vector $\textbf{p}=(p_1,\ldots,p_n)$ determines the
average throughput and power consumption of each node.

\section{Problem Formulation}\label{Problem Formulation}
\subsection{SINR Capture Model}\label{SINR_capture_model}
We first consider the probability of successful reception of the RTS packet from a particular node (thus the data transmission is granted to that node).
In a given handshake phase, the SINR of node $i$'s RTS packet is given by
\begin{align}
SINR_i=\frac{B_i|h_i|^2P_T}{N_0+\sum_{j\neq{}i}B_j|h_j|^2P_T}
\end{align}
where $B_i$ is a binary indicator which is $1$ if node $i$ sends an RTS in that handshake phase, and $0$ otherwise.
$N_0$ is the power of the additive noise at the BS, $h_i$ is the channel gain between node $i$ and the BS.
We assume $|h_i|^2$, $i=1,\ldots,n$, are independent, exponentially distributed random variables with mean one.

When $s$ nodes simultaneously transmit RTS packets to the BS, the probability of data transmission granted to a particular node (say, node $1$) is given by
\begin{align*}
Pr\left[SINR_1>b\right]
=Pr\left[\sum^s_{i=2}|h_i|^2-\frac{|h_1|^2}{b}<-\frac{N_0}{P_T}\right]
=\left(\frac{1}{1+b}\right)^{s-1}e^{-b\frac{N_0}{P_T}},
\end{align*}
where the last equality is obtained as follows.
Let the probability density function (PDF) of $|h_1|^2$ be $f_{|h_1|^2}(x)=e^{-x}, x\ge 0$, then $f_{-\frac{|h_1|^2}{b}}(x)=be^{bx}, x\le 0$, and the PDF of the sum of the independent and identically distributed (i.i.d.) $|h_j|^2$ is given by $f_{\sum^s_{i=2}|h_i|^2}(x)=\frac{x^{s-2}}{(s-2)!}e^{-x}, x\ge 0$. It follows that the PDF of
$\sum^s_{i=2}|h_i|^2-\frac{|h_1|^2}{b}$ for $x\le 0$ is given by
\begin{align*}
f_{\sum^s_{i=2}|h_i|^2-\frac{|h_1|^2}{b}}(x)=\int^{\infty}_0e^{-t}\frac{t^{s-2}}{(s-2)!}be^{b(x-t)}dt
=\left(\frac{1}{1+b}\right)^{s-1}be^{bx}~(\text{for}~x\le 0).
\end{align*}
Therefore,
\begin{align*}
Pr\left[\sum^s_{i=2}|h_i|^2-\frac{|h_1|^2}{b}<-\frac{N_0}{P_T}\right]
=\int^{-\frac{N_0}{P_T}}_{-\infty}f_{\sum^s_{i=2}|h_i|^2-\frac{|h_1|^2}{b}}(x) dx
=\left(\frac{1}{1+b}\right)^{s-1}e^{-b\frac{N_0}{P_T}}.
\end{align*}
Considering a given request probability vector $\textbf{p}=(p_1,\ldots,p_n)$ based on which the nodes send RTS packets, the probability of data transmission granted to a particular node can then be expressed as a function of $\textbf{p}$ by the following proposition.

\begin{proposition}
\label{Prob_Granted}
Assuming that the capture ratio is $b$, and there are $n$ nodes in the network having the request probability vector $\textbf{p}=(p_1,\ldots,p_n)$, then in a handshake phase, node $i$ is granted data transmission with probability
\begin{align}
\label{Eq_G_i}
G_i(p_1,\ldots,p_n)&\triangleq\ \mbox{Pr (BS grants node $i$ data transmission$|$\textbf{p})}\notag\\
&=e^{-b\frac{N_0}{P_T}}p_i\prod_{j\neq{}i}\left(1-\frac{bp_j}{1+b}\right).
\end{align}
\end{proposition}
\begin{proof}
Let $(x_1,\ldots,x_k)\in{}I_{-\{y_1,\ldots,y_l\}}$ denote $x_1<\cdots<x_k$, all belonging to the node index set $I_{-\{y_1,\ldots,y_l\}}\triangleq\{1,2,\ldots,n\}\setminus\ \{y_1,\ldots,y_l\}$, where $\setminus$ denotes the set minus operator.
Then
\begin{align*}
G_i(p_1,\ldots,p_n)=&\sum^{n}_{s=1}\mbox{Pr($s$ nodes request)$\cdot$Pr(node $i$ is granted$|$$s$ nodes request)}\\
=&p_i\cdot\prod_{j\in{}I_{-i}}(1-p_j)\cdot e^{-b\frac{N_0}{P_T}}\\
&+p_i\cdot\sum_{j\in{}I_{-i}}\left(p_j\prod_{k\in{}I_{-\{i,j\}}}(1-p_k)\right)\cdot\left(\frac{1}{1+b}\right)e^{-b\frac{N_0}{P_T}}\\
&+p_i\cdot\sum_{(j,k)\in{}I_{-i}}\left(p_jp_k\prod_{l\in{}I_{-\{i,j,k\}}}(1-p_l)\right)\cdot\left(\frac{1}{1+b}\right)^2e^{-b\frac{N_0}{P_T}}\\
&+\cdots\\
&+p_i\left(\prod_{j\in{}I_{-i}}p_j\right)\cdot\left(\frac{1}{1+b}\right)^{n-1}e^{-b\frac{N_0}{P_T}}\\
=&e^{-b\frac{N_0}{P_T}}p_i\left\{\prod_{j\in{}I_{-i}}(1-p_j)\right.\\
&+\sum_{j\in{}I_{-i}}\left(\left(\frac{p_j}{1+b}\right)\prod_{k\in{}I_{-\{i,j\}}}(1-p_k)\right)\\
&+\sum_{(j,k)\in{}I_{-i}}\left(\left(\frac{p_j}{1+b}\right)\left(\frac{p_k}{1+b}\right)\prod_{l\in{}I_{-\{i,j,k\}}}(1-p_l)\right)\\
&+\cdots\\
&\left.+\prod_{j\in{}I_{-i}}\left(\frac{p_j}{1+b}\right)\right\}\\
=&e^{-b\frac{N_0}{P_T}}p_i\prod_{j\neq{}i}\left[\left(\frac{p_j}{1+b}\right)+\left(1-p_j\right)\right]\\
=&e^{-b\frac{N_0}{P_T}}p_i\prod_{j\neq{}i}\left(1-\frac{bp_j}{1+b}\right).
\end{align*}
\end{proof}
In every handshake phase, the BS grants data transmission to node $i$ with probability $G_i$. It follows that, on average, node $i$ transmits data with period $G_iT_i$ after every handshake period $T_0$.

For the data transmission, assume that the entire data transmission period is encoded as one codeword which is called a \emph{frame}. Further assume that a good channel code, such as turbo codes, is used.
Then, the frame error rate at reasonable operating signal to noise ratios (SNR) is smaller or stays roughly the same as $T_i$ (code block size) increases \cite{turbo_blocksize98}\cite{turbo_blocksize01}. In time varying channels, the frame error rate decreases with $T_i$ even more evidently if proper interleaving is applied to exploit the increased time diversity (due to increased code block size).
We denote the frame success rate of node $i$ averaged over all possible channel realizations, when node $i$'s RTS packet is successfully received and its data transmission period is $T_i$, as $P^s_i(T_i)$. Note that in time-correlated channels with coherence time larger than the handshake period $T_0$, $P^s_i(T_i)$ is usually close to one due to node $i$'s good channel quality that won the competition in the handshake period.
Then, we have the average throughput as the following expression (this simple result can be formally obtained from the renewal process \cite{Renewal06}).
\begin{proposition}\label{Average_throughput}
The average throughput of node $i$ is given by
\begin{align}
r_i(\textbf{p})=\frac{P^s_i(T_i)G_iT_i}{T_0+\sum_jG_jT_j},
\end{align}
where $G_i=e^{-b\frac{N_0}{P_T}}p_i\prod_{j\neq i}(1-\frac{bp_j}{1+b})$, $P^s_i(T_i)$ is the average frame success rate of node $i$ when the data transmission period is $T_i$ slots, and we have used $\sum_j$ to denote $\sum^n_{j=1}$ for simplicity.
\end{proposition}

Let $S_i(\textbf{p})$ denote the normalized average power consumption of node $i$ (normalized by the transmission power $P_T$). Then $S_i(\textbf{p})$ is equal to the fraction of time in which node $i$ transmits either RTS or data packets.
In the sequel, we will simply call $S_i(\textbf{p})$ the average power consumption of node $i$ for brevity.
By defining $\tilde{T}_0<T_0$ as the actual duration of an RTS packet, the following proposition can be easily obtained from \emph{Proposition \ref{Average_throughput}}.

\begin{proposition}
The (normalized) average power consumption of node $i$ is given by
\begin{align}
\label{Eq_Si}
S_i(\textbf{p})=\frac{p_i\tilde{T}_0+G_iT_i}{T_0+\sum_jG_jT_j},
\end{align}
where $\tilde{T}_0<T_0$ is the actual duration of an RTS packet.
\end{proposition}

\subsection{Noncooperative Game Formulation}\label{game_formulation}
We use the concept of Nash equilibrium in game theory to formulate our problem.
The system can be modeled as a noncooperative game with constraints which are the average throughput demands. The nodes are the players, and the actions of a player (node) are: (i) selecting a request probability $p_i$ that can sustain the average throughput demand $\rho_i$ while minimizing the average power consumption $S_i$ in (\ref{Eq_Si}); (ii) transmitting an RTS packet with probability $p_i$ in every handshake phase. Note that action (i) has an action space $\{p_i: 0\leq p_i \leq 1\}$, while when $p_i$ has been chosen, action (ii) has only one element in the action space, that is to randomly transmit an RTS packet with probability $p_i$ in every handshake phase. Thus the Nash equilibrium will be analyzed with respect to the strategy for action (i), that is, what request probability to choose.
A Nash equilibrium point is a situation in which each node chooses its best strategy unilaterally to maximize its utility function (or minimize its cost function). The interested reader are referred to \cite{GameTheory94}\cite{GameTheory91} for further information about game theory.

Let $\textbf{p}_{-i}$ represent the vector of the request probabilities of all nodes except node $i$, and $r_i(p_i,\textbf{p}_{-i})$ represent the average throughput of node $i$ when it requests with probability $p_i$ given that the other nodes request with probability vector $\textbf{p}_{-i}$.
We define the utility function for node $i$ as $U_i(p_i,\textbf{p}_{-i})=1-S_i(\textbf{p})$ (which may be seen as the power left for node $i$), and the (constrained) Nash equilibrium point for our problem as follows.

\begin{definition}\label{def_NE}
A vector of the request probabilities $\textbf{p}$ is a (constrained) Nash equilibrium point if for all $i=1,\ldots,n$, we have
\begin{align}
\begin{cases}
r_i(p_i,\textbf{p}_{-i})\geq\rho_i\\
U_i(p_i,\textbf{p}_{-i})\ge U_i(\tilde{p}_i,\textbf{p}_{-i}),~\forall\tilde{p}_i\in\{\tilde{p}_i:r_i(\tilde{p}_i,\textbf{p}_{-i})\geq\rho_i\},
\end{cases}
\end{align}
where $\rho_i$, the average throughput demand, defines a constraint. 

Equivalently, $\textbf{p}$ is a Nash equilibrium point if
\begin{align}
p_i\in\arg\min_{0\le \tilde{p}_i\le 1}\{S_i(\tilde{p}_i,\textbf{p}_{-i}):r_i(\tilde{p}_i,\textbf{p}_{-i})\geq\rho_i\},~\forall i.
\end{align}
\end{definition}
The above expression means that at a Nash equilibrium point $\textbf{p}$, each node $i$ would not prefer to deviate from its choice of request probability.
It should be noted that our problem is a game with constraints,
so there are additional constraints in defining our Nash equilibrium point that differs from the conventional Nash equilibrium point.

Since $G_i(\textbf{p})$ is increasing in $p_i$, and decreasing in $p_j$ for $j\neq i$,
both the average throughput $r_i(\textbf{p})$ and the average power consumption $S_i(\textbf{p})$ are increasing in $p_i$. It follows that $(p_1,\ldots,p_n)$ (where $p_i \in [0,1], \forall i$) is a Nash equilibrium point if and only if it is a solution to the set of equations
\begin{align}\label{Equilibrium_equations}
r_i(\textbf{p})=\frac{P^s_i(T_i)G_iT_i}{T_0+\sum_jG_jT_j}=\rho_i,\ \forall i.
\end{align}

\emph{Remark}: The idea of Nash equilibrium point in game theory is from noncooperative interaction between nodes. Therefore, for most cases at the Nash equilibrium point, the system performance is suboptimal as compared to that at the traditional system-optimal solution.
In this paper, because the traditional system-optimal solution to satisfy the throughput demands $(\rho_1,\ldots,\rho_n)$ must also satisfy (\ref{Equilibrium_equations}), it is also a Nash equilibrium point defined above.

\section{Analysis of the Network}\label{Optimum Transmission Periods}
We now analyze the equilibrium equations. For conciseness of the derivation, let $\hat{\rho}_i \triangleq \frac{\rho_i}{P^s_i(T_i)}$. Taking summation of both sides of (\ref{Equilibrium_equations}), we have
\begin{align}
\label{Eq_rho_t}
&\frac{\sum_iG_iT_i}{T_0+\sum_jG_jT_j}=\sum_i\hat{\rho}_i\triangleq \rho_t\\
\Rightarrow\ &\sum_jG_jT_j=\frac{\rho_tT_0}{1-\rho_t}.\notag
\end{align}
Substituting this into (\ref{Equilibrium_equations}),
\begin{equation}
G_i=\frac{T_0\hat{\rho}_i}{T_i(1-\rho_t)}.
\label{Relation G_i and rho_i}
\end{equation}
Using \emph{Proposition \ref{Prob_Granted}}, we have the following proposition.
\begin{proposition}\label{Proposition NE equations}
Given the throughput demands $(\rho_1,\ldots,\rho_n)$, and let $\hat{\rho}_i \triangleq \frac{\rho_i}{P^s_i(T_i)}$ and $\sum_i\hat{\rho}_i\triangleq \rho_t$. The request probability vector $\textbf{p}=(p_1,\ldots,p_n)$ is a Nash equilibrium point if and only if
\begin{align}\label{Corresponding NE Equations_Extension}
\frac{T_0\hat{\rho}_i}{T_i(1-\rho_t)}=e^{-b\frac{N_0}{P_T}}p_i\prod_{j\neq{}i}\left(1-\frac{bp_j}{1+b}\right),\ \forall i.
\end{align}
\end{proposition}

Again, for conciseness of the derivation, let $\hat{T}_0 \triangleq e^{b\frac{N_0}{P_T}}T_0$. Hereafter
(\ref{Corresponding NE Equations}) will be used.
\begin{align}\label{Corresponding NE Equations}
\frac{\hat{T}_0\hat{\rho}_i}{T_i(1-\rho_t)}=p_i\prod_{j\neq{}i}\left(1-\frac{bp_j}{1+b}\right),\ \forall i.
\end{align}

\subsection{Feasible Throughput Region}
\begin{definition}
A throughput demand vector $(\rho_1,\ldots,\rho_n)$ is called feasible if
there is a Nash equilibrium $\textbf{p}=(p_1,\ldots,p_n)$ for it, that is, there is a solution to (\ref{Corresponding NE Equations}).
The feasible throughput region when node $i$ uses data transmission period $T_i$ is defined as
\begin{equation*}
\Omega(T_1,\ldots,T_n)\triangleq\{(\rho_1,\ldots,\rho_n):\mbox{a Nash equilibrium exists for}\ (\rho_1,\ldots,\rho_n)\}.
\end{equation*}
\end{definition}

We will show that if there is a Nash equilibrium point for $(\rho_1,\ldots,\rho_n)$ when the data transmission periods are $(T_1,\ldots,T_n)$, then there is also a Nash equilibrium point for $(\rho_1,\ldots,\rho_n)$ when the data transmission periods are $(T'_1,\ldots,T'_n)$, where $T'_i\ge T_i$ for all $i$. In other words, we have
\begin{proposition}
\label{Proof Optimum data transmission periods}
If $T'_i\ge T_i\ \forall\ i$, then $\Omega(T_1,\ldots,T_n)\subset\Omega(T'_1,\ldots,T'_n)$.
\end{proposition}
\begin{proof}
Assume $(\rho_1,\ldots,\rho_n)\in\Omega(T_1,\ldots,T_n)$, i.e., there is a request probability vector $(p^{(0)}_1,\ldots,p^{(0)}_n)$ satisfying
\begin{align*}
\frac{\hat{T}_0\hat{\rho}_i}{T_i(1-\rho_t)}=p^{(0)}_i\prod_{j\neq{}i}\left(1-\frac{bp^{(0)}_j}{1+b}\right),\ \forall i.
\end{align*}
We will start from the request probability vector $(p^{(0)}_1,\ldots,p^{(0)}_n)$, and successively update the request probability vector to $(p^*_1,\ldots,p^*_n)$ such that
\begin{align}\label{throughput_region_proof}
\frac{\hat{T}_0\hat{\rho}_i}{T'_i(1-\rho_t)}=p^*_i\prod_{j\neq{}i}\left(1-\frac{bp^*_j}{1+b}\right),\ \forall i.
\end{align}
Note that $\hat{\rho}_i = \frac{\rho_i}{P^s_i(T_i)}$, and as discussed in Section~\ref{SINR_capture_model}, $P^s_i(T'_i) \geq P^s_i(T_i)$ when $T'_i \geq T_i, \forall i$. If we choose $(\rho'_1,\ldots,\rho'_n)$ such that $\hat{\rho}_i = \frac{\rho_i}{P^s_i(T_i)}= \frac{\rho'_i}{P^s_i(T'_i)}$, then we will have $\rho'_i \geq \rho_i, \forall i$. If a solution $(p^*_1,\ldots,p^*_n)$ to (\ref{throughput_region_proof}) exists, we can conclude that $(\rho'_1,\ldots,\rho'_n)\in\Omega(T'_1,\ldots,T'_n)$, that is, there is a Nash equilibrium for the throughput demands $(\rho'_1,\ldots,\rho'_n)$. Since the average throughput $r_i(\textbf{p})$ is increasing in $p_i$, in this case, it can be easily verified that there will also be a Nash equilibrium for $(\rho_1,\ldots,\rho_n)$ when $\rho_i \leq \rho'_i, \forall i$. Therefore, we have proved that $(\rho_1,\ldots,\rho_n)\in\Omega(T'_1,\ldots,T'_n)$.

We now prove (\ref{throughput_region_proof}). Assume that $T_i<T'_i$ for some $i$ (otherwise, we are done since $T_i=T'_i$ for all $i$). Without loss of generality, assume that $i=1$.
Note that $p_i\prod_{j\neq{}i}\left(1-\frac{bp_j}{1+b}\right)$ decreases as $p_i$ decreases and it increases as $p_j$, $j\neq i$, decreases. There exists $p^{(1)}_1 < p^{(0)}_1$ such that
\begin{align*}
\frac{\hat{T}_0\hat{\rho}_1}{T'_1(1-\rho_t)}=p^{(1)}_1\prod_{j\neq{}1}\left(1-\frac{bp^{(0)}_j}{1+b}\right).
\end{align*}
We update the request probability vector from $(p^{(0)}_1,\ldots,p^{(0)}_n)$ to $(p^{(1)}_1,p^{(0)}_2,\ldots,p^{(0)}_n)$.
Then, by fixing the request probabilities of all nodes but node $2$, there exists $p^{(1)}_2 < p^{(0)}_2$ such that
\begin{align*}
\frac{\hat{T}_0\hat{\rho}_2}{T'_2(1-\rho_t)}=p^{(1)}_2\left(1-\frac{bp^{(1)}_1}{1+b}\right)\prod^n_{j=3}\left(1-\frac{bp^{(0)}_j}{1+b}\right).
\end{align*}
The request probability vector is updated from $(p^{(1)}_1,p^{(0)}_2,\ldots,p^{(0)}_n)$ to $(p^{(1)}_1,p^{(1)}_2,p^{(0)}_3,\ldots,p^{(0)}_n)$. 
This process is repeated to update the request probabilities to $p^{(1)}_j<p^{(0)}_j$, for $j\ge 3$, each time by fixing the request probabilities of all nodes but node $j$, until the request probability vector becomes $(p^{(1)}_1,p^{(1)}_2,\ldots,p^{(1)}_n)$.

We then consider again the request probability of node $1$. There exist $p^{(2)}_1<p^{(1)}_1$ such that
\begin{align*}
\frac{\hat{T}_0\hat{\rho}_1}{T'_1(1-\rho_t)}=p^{(2)}_1\prod_{j\neq{}1}\left(1-\frac{bp^{(1)}_j}{1+b}\right).
\end{align*}
Repeating the same process, we will have the updated request probability vector $(p^{(2)}_1,p^{(2)}_2,\ldots,p^{(2)}_n)$.
By continuously updating the request probability vector, we will get decreasing sequences $p^{(0)}_i > p^{(1)}_i > p^{(2)}_i > \cdots$, for $i=1,\ldots,n$. Because the request probabilities are lower bounded by zero, $p^{(q)}_i\rightarrow\ p^*_i$ as $q\rightarrow\infty$ for $i=1,\ldots,n$. It follows that
\begin{align*}
\frac{\hat{T}_0\hat{\rho}_i}{T'_i(1-\rho_t)}=p^*_i\prod_{j\neq{}i}\left(1-\frac{bp^*_j}{1+b}\right),\ \forall i.
\end{align*}
\end{proof}


We now give some properties about the feasible throughput region.
\begin{theorem}\label{Thm_NE}
There are at most two Nash equilibrium points for any feasible throughput demands $(\rho_1,\ldots,\rho_n)$, and exactly one Nash equilibrium point, called the \emph{better} Nash equilibrium point, with $\sum^n_{i=1} p_i\le \frac{b+1}{b}$. In other words, given any feasible throughput demands $(\rho_1,\ldots,\rho_n)$ in the feasible throughput region, there is exactly one Nash equilibrium $(p_1,\ldots,p_n)$ with $\sum^n_{i=1}p_i\le\frac{b+1}{b}$ satisfying (\ref{Corresponding NE Equations}).
\end{theorem}
\begin{proof}
In Appendix \ref{Proof_Thm_NE}.
\end{proof}
\emph{Remark}:
The better Nash equilibrium point is clearly the traditional system-optimal solution.

With \emph{Theorem \ref{Thm_NE}}, the following characterization of the feasible throughput region is easily obtained.
\begin{corollary}
\begin{align}
\Omega(T_1,\ldots,T_n)
=\left\{\left(\frac{P^s_1(T_1)G_1T_1}{T_0+\sum_jG_jT_j},\ldots,\frac{P^s_n(T_n)G_nT_n}{T_0+\sum_jG_jT_j}\right):
0\le p_i\le 1, \forall i,\sum_ip_i\le\frac{b+1}{b}\right\},
\end{align}
\end{corollary}
where $G_i=e^{-b\frac{N_0}{P_T}}p_i\prod_{j\neq i}(1-\frac{bp_j}{1+b})$.


\subsection{Power Consumption}
The following proposition gives the average power consumption at a Nash equilibrium point $\textbf{p}=(p_1,\ldots,p_n)$ for the throughput demands $(\rho_1,\ldots,\rho_n)$.
\begin{proposition}
The average power consumption of node $i$ at a Nash equilibrium point $\textbf{p}=(p_1,\ldots,p_n)$ for the feasible throughput demands $(\rho_1,\ldots,\rho_n)$ is given by
\begin{align}\label{Average_power_investment}
S_i(\textbf{p})=\hat{\rho}_i+\frac{\tilde{T}_0}{T_0}(1-\rho_t)p_i,
\end{align}
where $\tilde{T}_0<T_0$ is the actual duration of an RTS packet, $\hat{\rho}_i = \frac{\rho_i}{P^s_i(T_i)}$, and $\rho_t=\sum_i\hat{\rho}_i$.
\end{proposition}
\begin{proof}
In each time slot, the channel is in either the handshake phase or the transmission phase. Hence, at the Nash equilibrium point for the throughput demands $(\rho_1,\ldots,\rho_n)$, the fraction of time slots node $i$ transmits data equals to $\hat{\rho}_i$ (since the average frame success rate is $P^s_i(T_i)$), and the RTS/CTS handshake phase occupies a fraction $1-\sum_i\hat{\rho}_i$ of the total time slots. With node $i$'s request probability
$p_i$, the fraction of time in which node $i$ transmits RTS packets is $(1-\rho_t)p_i\tilde{T}_0/T_0$. The proposition follows since the average power consumption equals to the fraction of time node $i$ transmits RTS or data packets.
\end{proof}

Finally, we relate the total average power consumption to feasible throughput demands in an optimal system.
The key idea is in the following proposition:
\begin{proposition}
\label{Proposition_Equivalence_SumSi}
If node $i$ uses data transmission period $T_i$ for all $i$, then, within the feasible throughput region $\Omega(T_1,\ldots,T_n)$, the maximum total average power consumption $\sum_i S_i(\textbf{p})$ at the better Nash equilibrium point
is equal to the maximum total average power consumption in the region $\{(p_1,\ldots,p_n): \sum_ip_i\le \frac{b+1}{b}, 0\le p_i\le 1\}$.
\end{proposition}
\begin{proof}
The result follows directly from \emph{Theorem \ref{Thm_NE}} and its corollary.
\end{proof}

%

The following theorem gives an upper bound on the total average power consumption when $b>2$, and all nodes use the same data transmission period.
\begin{theorem}\label{Thm_PowerConsumption}
Assuming that the capture ratio $b>2$ and all nodes use the same data transmission period $MT_0$ and the same channel code, then for any feasible throughput demands $(\rho_1,\ldots,\rho_n)$ (i.e., $(\rho_1,\ldots,\rho_n)\in \Omega(MT_0,\ldots,MT_0)$), the total average power consumption $\sum_i S_i(\textbf{p})$ at the better Nash equilibrium point is upper bounded by
\begin{enumerate}
\item If $n=1$, then $S_1(p_1)\le \frac{M'+\beta}{M'+1}$.
\item If $n>1$, then
\begin{align}
\begin{cases}
   \sum_i S_i(\textbf{p})\le \frac{M'+\beta}{M'+1} & \text{if } \beta\le\frac{M'b(1-\Psi_{b,n})}{1+M'+M'b(1-\Psi_{b,n})},\\
   \sum_i S_i(\textbf{p})\le \frac{M'\Psi_{b,n}+\beta\frac{b+1}{b}}{M'\Psi_{b,n}+1} & \text{if }\frac{M'b(1-\Psi_{b,n})}{1+M'+M'b(1-\Psi_{b,n})}\le\beta\le\frac{b}{b+1},\\
   \sum_i S_i(\textbf{p})\le \frac{M'\Gamma(n)+\beta\frac{b+1}{b}}{M'\Gamma(n)+1} & \text{if }\beta\ge\frac{b}{b+1},
\end{cases}
\end{align}
\end{enumerate}
where $M'=Me^{-b\frac{N_0}{P_T}}$, $\beta=\frac{\tilde{T}_0}{T_0}$ is the RTS fraction, $\Gamma(n)=\frac{b+1}{b}(1-\frac{1}{n})^{n-1}$, and
\begin{align}
\Psi_{b,n}=\frac{(n-1)(1+b)^2-nb}{b(1+b)(n-1)}\left[\frac{bn+n-b-2}{(1+b)(n-1)}\right]^{n-2}
\end{align}
\end{theorem}
\begin{proof}
In Appendix \ref{Proof_Thm_PowerConsumption}.
\end{proof}

From the proof of \emph{Theorem \ref{Thm_PowerConsumption}}, we know that the bound is tight, i.e., the equality of the total average power consumption given in \emph{Theorem \ref{Thm_PowerConsumption}} can be satisfied by a point in the feasible throughput region (with $(p_1,\ldots,p_n) \in$ \{$(1,0,\ldots,0)$, $(1,\frac{1}{(n-1)b},\ldots,\frac{1}{(n-1)b})$, $(\frac{b+1}{nb},\ldots,\frac{b+1}{nb})$, and their permutations\}) when all nodes use the same data transmission period and channel code. For the general case, we give the following upper bound on the total average power consumption.
\begin{corollary}
Assuming that the capture ratio $b>2$ and node $i$ uses data transmission period $T_i$, then for any feasible throughput demands $(\rho_1,\ldots,\rho_n)\in \Omega(T_1,\ldots,T_n)$, the total average power consumption $\sum_i S_i(\textbf{p})$ at the better Nash equilibrium point is upper bounded by
\begin{enumerate}
\item If $n=1$, then $S_1(p_1)\le \frac{\overline{m}'+\beta}{\overline{m}'+1}$.
\item If $n>1$, then
\begin{align}\label{general bound}
\begin{cases}
   \sum_i S_i(\textbf{p})\le \frac{\overline{m}'+\beta}{\overline{m}'+1} & \text{if } \beta\le\frac{M'b(1-\Psi_{b,n})}{1+M'+M'b(1-\Psi_{b,n})},\\
   \sum_i S_i(\textbf{p})\le \frac{\overline{m}'\Psi_{b,n}+\beta\frac{b+1}{b}}{\overline{m}'\Psi_{b,n}+1} & \text{if }\frac{M'b(1-\Psi_{b,n})}{1+M'+M'b(1-\Psi_{b,n})}\le\beta\le\frac{b}{b+1},\\
   \sum_i S_i(\textbf{p})\le \frac{\underline{m}'\Gamma(n)+\beta\frac{b+1}{b}}{\underline{m}'\Gamma(n)+1} & \text{if }\beta\ge\frac{b}{b+1},
\end{cases}
\end{align}
\end{enumerate}
where $M'$, $\beta$, $\Gamma(n)$, $\Psi_{b,n}$ are defined as in \emph{Theorem \ref{Thm_PowerConsumption}}, $\overline{m}'=\overline{m}e^{-b\frac{N_0}{P_T}}$ with $\overline{m}=\max_i\left\{\frac{T_i}{T_0}\right\}$ and $\underline{m}'=\underline{m}e^{-b\frac{N_0}{P_T}}$ with $\underline{m}=\min_i\left\{\frac{T_i}{T_0}\right\}$.
\end{corollary}

In particular, we have
\begin{align}
\sum_i S_i(\textbf{p})\le \max\left\{1,\beta\frac{b+1}{b}\right\}.
\end{align}
\begin{proof}
When $n=1$, the result follows directly from \emph{Theorem \ref{Thm_PowerConsumption}}.
When $n>1$, by (\ref{Eq_rho_t}) and (\ref{Average_power_investment}), we have
\begin{align}
\sum_iS_i(\textbf{p})=\frac{\beta T_0(\sum_ip_i)+\sum_iG_iT_i}{T_0+\sum_iG_iT_i},
\end{align}
where $\beta=\frac{\tilde{T}_0}{T_0}$ and $G_i=e^{-b\frac{N_0}{P_T}}p_i\prod_{j\neq i}(1-\frac{bp_j}{1+b})$.
\begin{enumerate}
\item[(i)] If $\beta(\sum_ip_i)\le 1$, we have
$
\sum_iS_i(\textbf{p})\le\frac{\beta \left(\sum_ip_i\right)+\overline{m}\sum_jG_j}{1+\overline{m}\sum_jG_j}.
$
\item[(ii)] If $\beta(\sum_ip_i)\ge 1$, we have
$
\sum_iS_i(\textbf{p})\le\frac{\beta \left(\sum_ip_i\right)+\underline{m}\sum_jG_j}{1+\underline{m}\sum_jG_j}.
$
\end{enumerate}
When $\beta\left(\frac{b+1}{b}\right) \le 1$, we always have case (i) for all values of $\sum_i p_i$.
In this case, let the maximum of the total average power consumption $\sum_iS_i(\textbf{p})$ at the better Nash equilibrium maximized over the feasible throughput demands $(\rho_1,\ldots,\rho_n)\in\Omega(T_1,\ldots,T_n)$ be $\widetilde{S}$.
By \emph{Proposition \ref{Proposition_Equivalence_SumSi}}, we have
\[
\widetilde{S}\le \max_{\left\{(p_1,\ldots,p_n):\sum p_i\le\frac{b+1}{b},0\le p_i\le 1\right\}}
\frac{\beta (\sum_ip_i)+\overline{m}\sum_jG_j}{1+\overline{m}\sum_jG_j}.
\]
Note that the right-hand side of the inequality equals to the maximum total average power consumption at the better Nash equilibrium point of the case when all nodes use the same data transmission period $\overline{m}T_0$.
By \emph{Theorem \ref{Thm_PowerConsumption}}, the first two inequalities in (\ref{general bound}) can be obtained.

When $\beta\left(\frac{b+1}{b}\right) \ge 1$, both case (i) and case (ii) can happen. If $\beta(\sum_ip_i)\le 1$ (i.e., case (i)), we know from the second inequality of (\ref{general bound}) and the proof of \emph{Theorem \ref{Thm_PowerConsumption}} that
\begin{equation*}
\sum_i S_i (\textbf{p}) \le \frac{\overline{m}'\Psi_{b,n}+\beta\frac{b+1}{b}}{\overline{m}'\Psi_{b,n}+1} \le \frac{\underline{m}'\Psi_{b,n}+\beta\frac{b+1}{b}}{\underline{m}'\Psi_{b,n}+1} \le \frac{\underline{m}'\Gamma(n)+\beta\frac{b+1}{b}}{\underline{m}'\Gamma(n)+1}.
\end{equation*}

In the case when $\beta(\frac{b+1}{b})\ge 1$ and $\beta(\sum_ip_i)\ge 1$, again, let the maximum of the total average power consumption $\sum_iS_i(\textbf{p})$ at the better Nash equilibrium maximized over the feasible throughput demands $(\rho_1,\ldots,\rho_n)\in\Omega(T_1,\ldots,T_n)$ be $\widetilde{S}$.
By \emph{Proposition \ref{Proposition_Equivalence_SumSi}}, we have
\[
\widetilde{S}\le\max_{\left\{(p_1,\ldots,p_n):\sum p_i\le\frac{b+1}{b},0\le p_i\le 1\right\}}
\frac{\beta (\sum_ip_i)+\underline{m}\sum_jG_j}{1+\underline{m}\sum_jG_j}.
\]
The right-hand side of the inequality equals to the maximum total average power consumption at the better Nash equilibrium point of the case when all nodes use the same data transmission period $\underline{m}T_0$.
By \emph{Theorem \ref{Thm_PowerConsumption}}, and together with the above result for the $\beta(\frac{b+1}{b})\ge 1$ and $\beta(\sum_ip_i)\le 1$ case, the last inequality in (\ref{general bound}) can be obtained.

%
\end{proof}

\section{Conclusion}\label{Conclusion}
In this paper, we analyzed the feasible throughput region and power consumption of a CSMA-based network with heterogenous nodes, where the MAC protocol is CSMA with the RTS/CTS handshake.
The feasible throughput region in this network was characterized, and an upper bound of the total power consumption was provided for any throughput demands in the feasible throughput region. Specifically, the upper bound is satisfied by one of three points in the feasible throughput region depending on the RTS fraction when the lengths of the data transmission periods for all nodes are equal.

\appendix
\subsection{\textbf{Proof of Theorem \ref{Thm_NE}}}
\label{Proof_Thm_NE}
Let $\alpha=\frac{b}{1+b}$ and $\tilde{\rho}_i=\frac{\hat{T}_0\rho_i}{P^s_i(T_i)T_i(1-\rho_t)}$, both being constants determined by the system parameters. To show that the system of equations in (\ref{Corresponding NE Equations}) have at most two solutions
is equivalent to showing that there are at most two solutions of $(p_1,\ldots,p_n)$ satisfying
\begin{align}\label{Throughput_i_SINR}
\tilde{\rho}_i=p_i\prod_{j\neq{}i}\left(1-\alpha{}p_j\right),\ \forall i.
\end{align}
In addition, we need to show that if a solution exists, there is exactly one solution with $\sum_ip_i\le \frac{b+1}{b}$.
Without loss of generality, assume $\tilde{\rho}_i=\max_j\{\tilde{\rho}_j\}$ and $\min_{j}\{\tilde{\rho}_j\}>0$ (note: the node with throughput demand $0$ transmits RTS packets with probability $0$, and can be excluded without affecting the proof).

By (\ref{Throughput_i_SINR}), we have
\begin{align}\label{Eq_p_jtop_i_1}
\frac{p_j}{1-\alpha{}p_j}=\frac{\tilde{\rho}_j}{\tilde{\rho}_i}\frac{p_i}{1-\alpha{}p_i}\\
\label{Eq_p_jtop_i_2}\Rightarrow{}p_j=\frac{\frac{\tilde{\rho}_j}{\tilde{\rho}_i}p_i}{1-\alpha{}p_i+\alpha{}\frac{\tilde{\rho}_j}{\tilde{\rho}_i}p_i}.
\end{align}
This means that once $p_i$ is determined, $p_j$ is uniquely determined at the Nash equilibrium point, and $p_j$ increases if $p_i$ increases. Taking logarithm and then differentiating with respect to $p_i$ on both sides of (\ref{Eq_p_jtop_i_1}), we have
\begin{align}\label{derivative_of_p_j}
\left(\frac{1}{p_j}+\frac{\alpha}{1-\alpha{}p_j}\right)\frac{dp_j}{dp_i}=\frac{1}{p_i}+\frac{\alpha}{1-\alpha{}p_i}.
\end{align}
Taking logarithm on both sides of (\ref{Throughput_i_SINR}), we have
\begin{align*}
\ln{}\tilde{\rho}_i=\ln{}p_i+\sum_{j\neq{}i}\ln{}[1-\alpha{}p_j]\triangleq{}g(p_i).
\end{align*}
Recall that $p_j$ can be seen as a function of $p_i$ by (\ref{Eq_p_jtop_i_2}) and note that $p_j\le p_i$ since $\tilde{\rho}_i=\max_j\{\tilde{\rho}_j\}$.

We will show that there exist one or two solutions for $p_i$ with $0\le p_i\le 1$ given any feasible throughput demands.
Specifically, we will show that $g(p_i)$ is a \emph{unimodal function} (i.e., having only one local maximum, and the point at which the maximum occurs is called the mode) in $p_i$. The derivative of $g(p_i)$ is given by
\begin{align*}
\frac{dg(p_i)}{dp_i}=\frac{1}{p_i}-\sum_{j\neq{}i}\left(\frac{\alpha}{1-\alpha{}p_j}\frac{dp_j}{dp_i}\right).
\end{align*}
Using (\ref{derivative_of_p_j}), it follows that
\begin{align*}
\frac{dg(p_i)}{dp_i}=\frac{1}{p_i}-\sum_{j\neq{}i}\left(\frac{1}{p_i}+\frac{\alpha}{1-\alpha{}p_i}\right)\alpha{}p_j.
\end{align*}
The function $g(p_i)$ is increasing if and only if $\frac{dg(p_i)}{dp_i}\geq{}0$, that is,
\begin{align*}
&\sum_{j\neq{}i}\alpha{}p_j\leq\frac{\frac{1}{p_i}}{\frac{1}{p_i}+\frac{\alpha}{1-\alpha{}p_i}}=1-\alpha{}p_i\\
&\iff{}\sum^n_{j=1}p_j\leq\frac{1}{\alpha}=\frac{b+1}{b}.
\end{align*}
Similarly, we have that $g(p_i)$ is decreasing if $\sum^n_{j=1}p_j\geq{}\frac{b+1}{b}$.
Since $\sum^n_{j=1}p_j$ is an increasing function in $p_i$ (recall that $p_j$ increases if $p_i$ increases, $\forall{}j$), $g(p_i)$ is a unimodal function. Also recall that $p_j, \forall j$, is uniquely determined by $p_i$. It follows that there are at most two solutions given any feasible throughput demands, and exactly one is with $\sum^n_{j=1}p_j\leq{}\frac{b+1}{b}$ (the other with $\sum^n_{j=1}p_j\geq{}\frac{b+1}{b}$ if there are two solutions).
In summary,
we can achieved any feasible throughput demands by the request probabilities satisfying $\sum^n_{i=1}p_i\leq{}\frac{b+1}{b}$.

\subsection{\textbf{Proof of Theorem \ref{Thm_PowerConsumption}}}
\label{Proof_Thm_PowerConsumption}
For the case $n=1$, the average power consumption $S_1(p_1)$ can be obtained by using (\ref{Eq_G_i}) and (\ref{Eq_Si}). And it is straightforward to see that the maximum of $S_1(p_1)$ occurs when $p_1 = 1$.

We will prove for the $n\ge 2$ case in the following.
By (\ref{Relation G_i and rho_i}) and (\ref{Corresponding NE Equations}), the Nash equilibrium point $(p_1,\ldots,p_n)$ has the following relation when data transmission periods $T_i=MT_0$ for all $i$
\begin{align*}
\hat{G}_i=p_i\prod_{j\neq{}i}\left(1-\frac{bp_j}{1+b}\right),\ \forall i,
\end{align*}
where $\hat{G}_i=\frac{\hat{\rho}_ie^{b\frac{N_0}{P_T}}}{M(1-\rho_t)}=G_i e^{b\frac{N_0}{P_T}}$ is defined to make the following proof concise.

We first give some lemmas required to complete the proof.
\begin{lemma}\label{Lemma bounds on sum G_i}
Given fixed $\sum_{i=1}^np_i=C$ at the better Nash equilibrium point (i.e., $0\le p_i\le 1$ and $C\le\frac{b+1}{b}$), then the minimum of $\sum^n_{i=1}\hat{G}_i$ can be achieved by $(p^*_1,\ldots,p^*_n)=(\frac{C}{n},\ldots,\frac{C}{n})$ and the maximum of $\sum^n_{i=1}\hat{G}_i$ can be achieved by one of the following points
\begin{enumerate}
\item when $C\le 1$: $(p^*_1,\ldots,p^*_n)\in$ \{$(C,0,\ldots,0)$ and its permutations\}
\item when $C\ge 1$: $(p^*_1,\ldots,p^*_n)\in$ \{$(1,C-1,0,\ldots,0)$, $(1,\frac{C-1}{2}$, $\frac{C-1}{2},0,\ldots,0),\ldots,$ $(1,\frac{C-1}{n-1},\ldots,\frac{C-1}{n-1})$, and their permutations\}
\end{enumerate}
\end{lemma}
\begin{proof}
We can treat $\sum_i\hat{G}_i$ as a function of $(p_1,\ldots,p_n)$, and find the critical points in the region $\{\sum^n_{i=1}p_i=C, 0<p_i<1\}$ by the \emph{Lagrange method}:
\begin{align*}
&\frac{\partial}{\partial{}p_i}\left\{\sum^n_{i=1}\hat{G}_i+\lambda\left(\sum^n_{i=1}p_i\right)\right\}=0\\
&\Rightarrow{}\left[\prod_{j\neq{}i}\left(1-\frac{b}{1+b}p_j\right)\right]\left(1-\sum_{j\neq{}i}\frac{\frac{b}{1+b}p_j}{1-\frac{b}{1+b}p_j}\right)+\lambda=0.
\end{align*}
Considering the partial derivatives for $i=1$ and $i=2$, we have
\begin{align*}
\left[\prod_{j\neq{}1}\left(1-\frac{b}{1+b}p_j\right)\right]\left(1-\sum_{j\neq{}1}\frac{\frac{b}{1+b}p_j}{1-\frac{b}{1+b}p_j}\right)
&=\left[\prod_{j\neq{}2}\left(1-\frac{b}{1+b}p_j\right)\right]\left(1-\sum_{j\neq{}2}\frac{\frac{b}{1+b}p_j}{1-\frac{b}{1+b}p_j}\right)\\
\iff~~~~~ \frac{1}{1-\frac{b}{1+b}p_1}\left[1-\sum_{j\neq{}1}\frac{\frac{b}{1+b}p_j}{1-\frac{b}{1+b}p_j}\right]
&=\frac{1}{1-\frac{b}{1+b}p_2}\left[1-\sum_{j\neq{}2}\frac{\frac{b}{1+b}p_j}{1-\frac{b}{1+b}p_j}\right].
\end{align*}
Let $\gamma_i=\frac{1}{1-\frac{b}{1+b}p_i}$. We have
\begin{align*}
&\gamma_1\left[\gamma_1+\sum^n_{j=1}(1-\gamma_j)\right]=\gamma_2\left[\gamma_2+\sum^n_{j=1}(1-\gamma_j)\right]\\
&\Rightarrow{}-\gamma_1\gamma_2+\gamma_1(n-\gamma_3-\cdots\gamma_n)=-\gamma_1\gamma_2+\gamma_2(n-\gamma_3-\cdots\gamma_n)\\
&\Rightarrow{}\gamma_1=\gamma_2\ \mbox{or}\ \sum^n_{i=3}\gamma_i=n .
\end{align*}
Similar results can be obtained by considering the partial derivatives with respect to any different $p_i$ and $p_j$. Therefore, we have
\begin{align}
&\gamma_i=\gamma_j\ \forall{i,j} \label{gamma_i}\\
&\mbox{or}~ \sum^n_{i=3}\gamma_i=n
\Rightarrow{}\sum^{n-1}_{i=2}\gamma_i=n\ \ (\mbox{by the argument of symmetry}) \notag\\
&\Rightarrow{}\gamma_2=\gamma_n \notag \\
&\Rightarrow{}\gamma_i=\gamma_j,\ \forall{i,j} \ \ (\mbox{by the argument of symmetry}) \label{gamma_ii}.
\end{align}
Both (\ref{gamma_i}) and (\ref{gamma_ii}) result in
$p_i=p_j=\frac{C}{n},\ \forall{i,j}$.
This shows that $\sum^n_{i=1}\hat{G}_i$ has only one critical point in the region $\{\sum^n_{i=1}p_i=C, 0 < p_i < 1\}$, with value
\[
\sum^n_{i=1}\hat{G}_i=C\left(1\!-\!\frac{b}{1+b}\frac{C}{n}\right)^{n-1}.
\]
It can be shown that this value is a minimum. Since there is only one critical point, the minimum can not occur on the boundary of the region. Hence $(p^*_1,\ldots,p^*_n)=(\frac{C}{n},\ldots,\frac{C}{n})$ achieves the global minimum.

For the maximum of $\sum^n_{i=1}\hat{G}_i$, we know that it must occur on the boundary of the region $\{(p_1,\ldots,p_{n}): \sum^n_{i=1}p_i=C, 0 \le p_i \le 1\}$ because there is only one critical point and the point is a minimum.
Since the problem is symmetric with respect to the nodes, in the following, we will only consider the representative solutions of $(p^*_1,\ldots,p^*_n)$. It is straightforward to see that their permutations are also solutions.

When $C \le 1$, the boundary is $\bigcup_{j=1}^{n}\{ (p_1,\ldots,p_{n}): p_j = 0, \sum^n_{i=1, i\neq j}p_i=C, p_i \ge 0, \forall i \neq j\}$.
Note that if some $p_i$'s are zeroes, then the corresponding nodes' throughputs are zero, and we can remove them and the problem is reduced to itself with fewer variables.
So if $(p^*_1,\ldots,p^*_n)$ achieves the maximum of $\sum^n_{i=1}\hat{G}_i$, it must always be on the boundary when we reduce the problem to another one with fewer variables.
It can then be easily seen that the boundary point $(p^*_1,\ldots,p^*_n)=(C,0,\ldots,0)$ achieves the maximum of $\sum^n_{i=1}\hat{G}_i$, which can be verified since $\sum^n_{i=1}\hat{G}_i\le\sum^{n-1}_{i=1}p_i=C$.


When $C>1$, the boundary is $\bigcup_{j=1}^{n} \big\{$ $\{ (p_1,\ldots,p_{n}): p_j = 0, \sum^n_{i=1, i\neq j}p_i=C, 0 \le p_i \le 1, \forall i \neq j\}\bigcup$ $\{ (p_1,\ldots,p_{n}): p_j = 1, \sum^n_{i=1, i\neq j}p_i=C-1, p_i \ge 0, \forall i \neq j\}$ $\big\}$.
We have the set of boundary points $\{(\frac{C}{2},\frac{C}{2},0,\ldots,0),\ldots,$ $(\frac{C}{n-1},\ldots,\frac{C}{n-1},0),$ and their permutations$\}$
which are critical points that achieve the minimum $\sum^n_{i=1}\hat{G}_i$ (of the corresponding problem dimensions) as we remove the zero-throughput nodes to reduce the problem.
Taking the reduced problems as well as the original problem into consideration, we can see that the the maximum of $\sum^n_{i=1}\hat{G}_i$ will eventually occur on the boundary $\bigcup_{j=1}^{n} \{ (p_1,\ldots,p_{n}): p_j = 1, \sum^n_{i=1, i\neq j}p_i=C-1, p_i \ge 0, \forall i \neq j\}$.
Without loss of generality,
we consider the boundary points $\{ (p_1,\ldots,p_{n-1},1): \sum^{n-1}_{i=1}p_i=C-1, p_i \ge 0, i \neq n\}$. Similar to the derivation of the minimum of $\sum^n_{i=1}\hat{G}_i$, it can be shown that $\sum^n_{i=1}\hat{G}_i$ has only one critical point $(\frac{C-1}{n-1},\ldots,\frac{C-1}{n-1},1)$, which is the maximum, in the region $\{ (p_1,\ldots,p_{n-1},1): \sum^{n-1}_{i=1}p_i=C-1, p_i > 0, i \neq n\}$. 
Therefore, when $C>1$
we only need to consider $(p^*_1,\ldots,p^*_n)\in$ \{$(1,C-1,0,\ldots,0)$, $(1,\frac{C-1}{2}$, $\frac{C-1}{2},0,\ldots,0),\ldots,$ $(1,\frac{C-1}{n-1},\ldots,\frac{C-1}{n-1})$ and their permutations\}
for the maximum $\sum^n_{i=1}\hat{G}_i$.
The proof is complete.
\end{proof}

%
%

\begin{lemma}\label{Increase_rho_increase_S}
In a network consisting of $n$ homogeneous nodes having the same data transmission period $T_i=MT_0, \forall i$, the same channel code and feasible throughput demands $(\rho,\ldots,\rho)$, we have the average power consumption at the better Nash equilibrium point increase as $\rho$ increases.
\end{lemma}
\begin{proof}
By (\ref{Corresponding NE Equations}), we have the following relation for any feasible throughput demands $(\rho,\ldots,\rho)$:
\begin{align}\label{Rho_in_p}
&\frac{\hat{\rho}}{(1-n\hat{\rho})M'}=p\left(1-\frac{bp}{1+b}\right)^{n-1}\notag\\
\Rightarrow~&\hat{\rho}=\frac{M'p(1-\alpha p)^{n-1}}{1+nM'p(1-\alpha p)^{n-1}},
\end{align}
where $M'=Me^{-b\frac{N_0}{P_T}}$, $\hat{\rho}=\frac{\rho}{P^s}$, $P^s$ is the average frame success rate when the data transmission period is $MT_0$, $\alpha=\frac{b}{1+b}$ and $(p,\ldots,p)$ is the request probability vector.

It can be easily verified that the maximum of $\rho$ is achieved when $p=\frac{1}{n\alpha}$, and the value of $\rho$ increases as $p$ increases when $0\le p \le \frac{1}{n\alpha}$. That is, the feasible throughput demand $\rho$ is an increasing function in $p$ at the better Nash equilibrium point (recall that $np\le \frac{1}{\alpha}$ at the better Nash equilibrium point by \emph{Theorem \ref{Thm_NE}}).


It remains to show that the average power consumption given by
\[
S_i(p)=\hat{\rho}+\beta(1-n\hat{\rho})p
\]
is an increasing function in $p$ for $0\le p\le \frac{1}{n\alpha}$, where $\beta=\frac{\tilde{T}_0}{T_0}$.
Equivalently, we will show that $\frac{dS_i(p)}{dp}\ge 0$ for $0\le p\le \frac{1}{n\alpha}$.

Substitute (\ref{Rho_in_p}) into $S_i(p)$
and then differentiate with respect to $p$. After some manipulations, we arrive at
\begin{equation*}
\frac{dS_i(p)}{dp}=\frac{M'(1-\alpha p)^n(1-\alpha pn+\alpha p^2\beta n^2-\beta\alpha p^2n)+\beta(1-\alpha p)^2}{\left[1-\alpha p+nM'p(1-\alpha p)^n\right]^2}.
\end{equation*}
For $0\le p\le \frac{1}{n\alpha}$, we have $0\le \alpha pn\le 1$ and $(\beta pn-\beta p-1)\ge -1$.
Thus
\begin{align*}
&1+ \alpha pn(\beta pn-\beta p-1)\ge 0\\
\iff &1-\alpha pn+\alpha p^2\beta n^2-\beta\alpha p^2n\ge 0.
\end{align*}
It can then be easily observed that $\frac{dS_i(p)}{dp}\ge 0$ for $0\le p\le \frac{1}{n\alpha}$.
\end{proof}

We now start the proof of \emph{Theorem \ref{Thm_PowerConsumption}} for $n>1$.
\begin{proof}[Proof of Theorem \ref{Thm_PowerConsumption}]\\
By (\ref{Eq_rho_t}) and (\ref{Average_power_investment}), we have the total average power consumption for throughput demands $(\rho_1,\ldots,\rho_n)$ as follows
\begin{align}\label{Sum_power}
\sum_i S_i(\textbf{p})=\frac{\beta (\sum_ip_i)+M\sum_iG_i}{1+M\sum_iG_i}.
\end{align}

We want to find the maximum of the total average power consumption given by (\ref{Sum_power}) among all feasible throughput demands at the better Nash equilibrium point, or equivalently among the set $\{\textbf{p}:0\le p_i\le 1\ \forall\ i, \sum_ip_i\le\frac{b+1}{b}\}$ by \emph{Proposition \ref{Proposition_Equivalence_SumSi}}.

First note that the function $f(x)=\frac{K+x}{1+x},x\ge 0$ is an increasing function if $K<1$, and a decreasing function if $K>1$. So for fixed $\sum_ip_i$ we have
\begin{enumerate}
\item[(i)] If $\beta(\sum_ip_i)\le 1$, $\sum_iS_i(\textbf{p})$ is maximized when $\sum_i G_i$ is maximized.
\item[(ii)] If $\beta(\sum_ip_i)\ge 1$, $\sum_iS_i(\textbf{p})$ is maximized when $\sum_i G_i$ is minimized.
\end{enumerate}

When $\beta(\frac{b+1}{b})\le 1$, we have case (i) for all values of $\sum_i p_i$ because $\sum_ip_i\le\frac{b+1}{b}$. For this case, it follows from \emph{Lemma \ref{Lemma bounds on sum G_i}} and $\hat{G}_i = G_i e^{b\frac{N_0}{P_T}}$ that we only need to consider
\begin{itemize}
\item when $C\le 1$: $(p_1,\ldots,p_n)\in$ \{$(C,0,\ldots,0)$ and its permutations\};
\item when $1\le C\le \frac{b+1}{b}$: $(p_1,\ldots,p_n)\in$ \{$(1,C-1,0,\ldots,0)$, $(1,\frac{C-1}{2}$, $\frac{C-1}{2},0,\ldots,0),\ldots,$ $(1,\frac{C-1}{n-1},\ldots,\frac{C-1}{n-1})$, and their permutations\}.
\end{itemize}

Due to the symmetry of the problem, we will only consider the representative $(p_1,\ldots,p_n)$'s. By (\ref{Sum_power}) and (\ref{Eq_G_i}), the maximum of $\sum_i S_i(\textbf{p})$ when $(p_1,\ldots,p_n)=(C,0,\ldots,0)$ with $\sum_ip_i=C\le1$ is clearly $\frac{M'+\beta}{M'+1}$ when $C=1$, where $M'=Me^{-b\frac{N_0}{P_T}}$.
For $1\le C\le\frac{b+1}{b}$,
the total average power consumption $\sum_i S_i(\textbf{p})$ when $(p_1,\ldots,p_n)=(1,x,\ldots,x)$, with $0\le x\le\frac{1}{(n-1)b}$,
is given by
\begin{equation}
\sum_i S_i((1,x,\ldots,x))=\frac{\beta[1+(n-1)x]+M'[1+(n-1)x-n\alpha x](1-\alpha x)^{n-2}}{1+M'[1+(n-1)x-n\alpha x](1-\alpha x)^{n-2}},
\label{average_power_consumption}
\end{equation}
where $\alpha=\frac{b}{1+b}$.

We first consider the point $(1,\frac{C-1}{n-1},\ldots,\frac{C-1}{n-1})$.
Define
\[
\Phi_{b,n}(x)=[1+(n-1)x-n\alpha x](1-\alpha x)^{n-2}  
\]
and
\begin{align}
\Psi_{b,n}=\Phi_{b,n}\left(\frac{1}{(n-1)b}\right)=\frac{(n-1)(1+b)^2-nb}{b(1+b)(n-1)}\left[\frac{bn+n-b-2}{(1+b)(n-1)}\right]^{n-2}.
\end{align}
From (\ref{average_power_consumption}), when $x=\frac{1}{(n-1)b}$,
\[
\sum_i S_i((1,x,\ldots,x))= \frac{M'\Psi_{b,n}+\beta\frac{b+1}{b}}{M'\Psi_{b,n}+1}.
\]
We will show that the maximum of $\sum_iS_i((1,x,\ldots,x))$ is either
$\frac{M'+\beta}{M'+1}$
when $x=0$, or
$\frac{M'\Psi_{b,n}+\beta\frac{b+1}{b}}{M'\Psi_{b,n}+1}$ when $x=\frac{1}{(n-1)b}$, that is,
\begin{align}\label{Eq_G_i_atboundaryPt}
\max_{0\le x\le\frac{1}{(n-1)b}}\sum_i S_i((1,x,\ldots,x))=\max\left\{\frac{M'+\beta}{M'+1},\frac{M'\Psi_{b,n}+\beta\frac{b+1}{b}}{M'\Psi_{b,n}+1}\right\}.
\end{align}

Note that for $\frac{M'+\beta}{M'+1}$ to be the maximum, we have
\begin{align*}
&\frac{M'+\beta}{M'+1}\ge\sum_iS_i((1,x,\ldots,x))\\
\Leftrightarrow~&\beta\le\frac{M'(1-\Phi_{b,n}(x))}{(M'+1)(1+(n-1)x)-(1+M'\Phi_{b,n}(x))}\triangleq\beta_1.
\end{align*}
Similarly, for $\frac{M'\Psi_{b,n}+\beta\frac{b+1}{b}}{M'\Psi_{b,n}+1}$ to be the maximum,
\begin{align*}
&\frac{M'\Psi_{b,n}+\beta\frac{b+1}{b}}{M'\Psi_{b,n}+1}\ge\sum_iS_i((1,x,\ldots,x))\\
\Leftrightarrow~ &\beta\ge\frac{M'(\Phi_{b,n}(x)-\Psi_{b,n})}{(M'\Phi_{b,n}(x)+1)\frac{b+1}{b}-(1+(n-1)x)(1+M'\Psi_{b,n})}\triangleq\beta_2.
\end{align*}
To prove (\ref{Eq_G_i_atboundaryPt}), we will show that $\beta_1\ge\beta_2$. In that case, if $\beta > \beta_1$ ($\frac{M'+\beta}{M'+1}$ is not the maximum), we will have $\beta > \beta_2$ ($\frac{M'\Psi_{b,n}+\beta\frac{b+1}{b}}{M'\Psi_{b,n}+1}$ is the maximum). Similarly, if $\beta < \beta_2$ ($\frac{M'\Psi_{b,n}+\beta\frac{b+1}{b}}{M'\Psi_{b,n}+1}$ is not the maximum), we will have $\beta < \beta_1$ ($\frac{M'+\beta}{M'+1}$ is the maximum).

We first show that $\Phi_{b,n}(x)$ is decreasing in $0\le x\le \frac{1}{(n-1)b}$.
\[
\frac{d}{dx}\Phi_{b,n}(x)\le 0 \Leftrightarrow \alpha(n\alpha-n+1)x\le (2\alpha-1)
\]
This inequality is satisfied if $(n\alpha-n+1)\le 0$, since $x\ge 0$ and $2\alpha-1=2\frac{b}{1+b}-1>\frac{1}{3}$.
If $(n\alpha-n+1)>0$, or equivalently, $n<b+1$, we need to show that $x\le \frac{2\alpha-1}{\alpha(n\alpha-n+1)}$.
This is satisfied because we have $n\ge 2>\frac{b+1}{b}$ and then
$
x\le \frac{1}{(n-1)b}< \frac{2\alpha-1}{\alpha(n\alpha-n+1)}.
$
Therefore, $\Phi_{b,n}(x)$ is decreasing in $0\le x\le \frac{1}{(n-1)b}$.

The fact that $\Phi_{b,n}(x)$ is decreasing in $0\le x\le \frac{1}{(n-1)b}$ implies that $(M'\Phi_{b,n}(x)+1)\frac{b+1}{b}-(1+(n-1)x)(1+M'\Psi_{b,n})$ is decreasing in $0\le x\le \frac{1}{(n-1)b}$, and
it follows that $(M'\Phi_{b,n}(x)+1)\frac{b+1}{b}-(1+(n-1)x)(1+M'\Psi_{b,n})\ge 0$ since its value is zero when $x=\frac{1}{(n-1)b}$. In addition, $(M'+1)(1+(n-1)x)-(1+M'\Phi_{b,n}(x))$ is increasing in $0\le x\le \frac{1}{(n-1)b}$, and $(M'+1)(1+(n-1)x)-(1+M'\Phi_{b,n}(x)) \ge 0$ since its value is zero when $x=0$. With these properties, we have
\begin{align*}
&\beta_1\ge\beta_2\\
\Leftrightarrow~& (1-\Phi_{b,n}(x))\left[(M'\Phi_{b,n}(x)+1)\frac{b+1}{b}-(1+(n-1)x)(1+M'\Psi_{b,n})\right]\\
&\ge (\Phi_{b,n}(x)-\Psi_{b,n})[(M'+1)(1+(n-1)x)-(1+M'\Phi_{b,n}(x))] \\
\Leftrightarrow~& \frac{1+M'\Phi_{b,n}(x)}{b}\left[bx(1-n)(1-\Psi_{b,n})+1-\Phi_{b,n}(x)\right]\ge 0.
\end{align*}
Let $U(x)=bx(1-n)(1-\Psi_{b,n})+1-\Phi_{b,n}(x)$. We have $U(0)=0$ and $U(\frac{1}{(n-1)b})=0$, so $\beta_1\ge\beta_2$ if the second derivative of $U(x)$, $\frac{d^2U}{dx^2}\le 0$ in $0\le x\le \frac{1}{(n-1)b}$.
\[
\frac{d^2U}{dx^2}\le 0 \Leftrightarrow \alpha(n-2)(n-n\alpha-1)x\ge (n-2)(2-3\alpha)
\]
Note that $2-3\alpha=2-3\frac{b}{1+b}< 0$. The above inequality is clearly satisfied if $n-n\alpha-1\ge 0$ (that is, $n \ge b+1$) or $n=2$, so we only need to show that $x\le\frac{2-3\alpha}{\alpha(n-n\alpha-1)}$ if $3\le n<b+1$.
We have
\[
x\le\frac{1}{(n-1)b}\le \frac{2-3\alpha}{\alpha(n-n\alpha-1)}~\Leftrightarrow~n\ge\frac{b^2-1}{b^2-b-1},
\]
which is satisfied because $n\ge3\ge\frac{b^2-1}{b^2-b-1}$.
Now, we have completed the proof of (\ref{Eq_G_i_atboundaryPt}).

For the points $(1,C-1,0,\ldots,0)$, $(1,\frac{C-1}{2}$, $\frac{C-1}{2},0,\ldots,0),\ldots,$ $(1,\frac{C-1}{n-2},\ldots,\frac{C-1}{n-2},0)$, the above proof still applies, and results similar to (\ref{Eq_G_i_atboundaryPt}) (with some nodes having zero request probabilities) can be derived.
We find that for $b>2$, $\Psi_{b,2}<\cdots<\Psi_{b,n}$, or equivalently,
$
\frac{M'\Psi_{b,2}+\beta\frac{b+1}{b}}{M'\Psi_{b,2}+1} < \cdots < \frac{M'\Psi_{b,n}+\beta\frac{b+1}{b}}{M'\Psi_{b,n}+1}.
$
Therefore, we conclude that when $\beta\frac{b+1}{b}\le1$, the maximum total average power consumption of the $n$-node network is either $\frac{M'+\beta}{M'+1}$ or $\frac{M'\Psi_{b,n}+\beta\frac{b+1}{b}}{M'\Psi_{b,n}+1}$. The inequality $\Psi_{b,i+1}-\Psi_{b,i}>0$ can be shown by transforming it into a polynomial of $b$, and showing that it is positive for $b>2$. Because it is difficult to mathematically prove this, we plot the polynomial and observe that the polynomial is positive for $b> \zeta_i$ (or $0<\frac{1}{b}<\frac{1}{\zeta_i}$ by using the corresponding polynomial with variable $b'=\frac{1}{b}$), where $\zeta_i\le 2$ and is decreasing in $i$, as shown in Fig.~\ref{Fig_Thd}.

Comparing $\frac{M'+\beta}{M'+1}$ and $\frac{M'\Psi_{b,n}+\beta\frac{b+1}{b}}{M'\Psi_{b,n}+1}$, we have
\begin{align*}
&\frac{M'\Psi_{b,n}+\beta\frac{b+1}{b}}{M'\Psi_{b,n}+1} \ge\frac{M'+\beta}{M'+1}\\
\iff &\beta\ge\frac{M'b(1-\Psi_{b,n})}{1+M'+M'b(1-\Psi_{b,n})}.
\end{align*}
Therefore, when $\beta(\frac{b+1}{b})\le 1$, the total average power consumption of the $n$-node network $\sum_i S_i(\textbf{p}) \le \frac{M'+\beta}{M'+1}$ if $\beta\le\frac{M'b(1-\Psi_{b,n})}{1+M'+M'b(1-\Psi_{b,n})}$;
and $\sum_i S_i (\textbf{p}) \le \frac{M'\Psi_{b,n}+\beta\frac{b+1}{b}}{M'\Psi_{b,n}+1}$ if $\frac{M'b(1-\Psi_{b,n})}{1+M'+M'b(1-\Psi_{b,n})}\le\beta\le\frac{b}{b+1}$.

When $\beta(\frac{b+1}{b})\ge 1$, both case (i) and case (ii) can happen. If $\beta(\sum_ip_i)\le 1$ (i.e., case (i)), we know from the above derivation that $\sum_i S_i (\textbf{p}) \le \frac{M'\Psi_{b,n}+\beta\frac{b+1}{b}}{M'\Psi_{b,n}+1}$ because $\beta \ge \frac{b}{b+1} > \frac{M'b(1-\Psi_{b,n})}{1+M'+M'b(1-\Psi_{b,n})}$. Now define $\Gamma(k) \triangleq \frac{b+1}{b}(1-\frac{1}{k})^{k-1}$.
By \emph{Lemma \ref{Lemma bounds on sum G_i}}, the minimum of $\sum_i \hat{G}_i$ when $\sum_ip_i= \frac{b+1}{b}$ is $\Gamma(n)$ which is achieved when $p_i=\frac{b+1}{nb}, \forall i$.
Since $\Psi_{b,n}$ is a realization of $\sum_i \hat{G}_i$, we have $\Gamma(n)\le \Psi_{b,n}$.
Therefore, when $\beta(\frac{b+1}{b})\ge 1$ and $\beta(\sum_ip_i)\le 1$ we have
$
\sum_i S_i (\textbf{p}) \le \frac{M'\Psi_{b,n}+\beta\frac{b+1}{b}}{M'\Psi_{b,n}+1} \le \frac{M'\Gamma(n)+\beta\frac{b+1}{b}}{M'\Gamma(n)+1}.
$

When $\beta(\frac{b+1}{b})\ge 1$ and $\beta(\sum_ip_i)\ge 1$ (i.e., case (ii)), it follows from \emph{Lemma \ref{Lemma bounds on sum G_i}} that the minimum of $\sum_i G_i$ is achieved when all $p_i$'s are equal. 
In addition, by \emph{Lemma \ref{Increase_rho_increase_S}}, the total average power consumption $\sum_i S_i(\textbf{p})$ at the better Nash equilibrium point in a homogeneous network is maximized by the maximum feasible throughput demands.
For $n$ homogenous nodes with throughput demands $(\rho,\ldots,\rho)$, we have the maximum feasible throughput $n\rho^{max}= \frac{M'\Gamma(n)}{1+M'\Gamma(n)}P^s$ (obtained when $p=\frac{1}{n\alpha}$ by (\ref{Rho_in_p})),
where $P^s$ is the average frame success rate when the data transmission period is $MT_0$. By (\ref{Sum_power}) with $\sum_i G_i = e^{-b\frac{N_0}{P_T}}\frac{b+1}{b}\left(1-\frac{1}{n}\right)^{n-1}$, the corresponding $\sum_i S_i(\textbf{p})$ is $\frac{M'\Gamma(n)+\beta\frac{b+1}{b}}{M'\Gamma(n)+1}$.

In summary, when $\beta(\frac{b+1}{b})\ge 1$,
$\sum_iS_i (\textbf{p}) \le \frac{M'\Gamma(n)+\beta\frac{b+1}{b}}{M'\Gamma(n)+1}$. In addition,
$1 \le \frac{M'\Gamma(n)+\beta\frac{b+1}{b}}{M'\Gamma(n)+1} \le \beta\frac{b+1}{b}$.

\end{proof}

\begin{figure}[!t]
\centering
\includegraphics[width=0.7\textwidth]{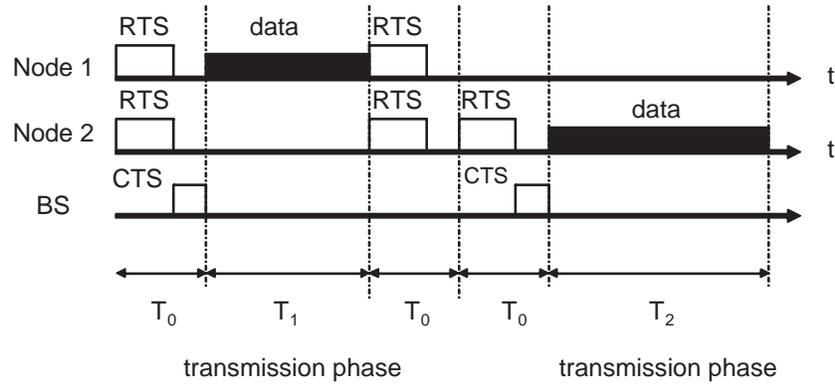}
\caption{Illustration of the CSMA by the RTS/CTS handshake mechanism in the network.}
\label{RTS_CTS}
\end{figure}

\begin{figure}[!t]
\centering
\includegraphics[width=0.7\textwidth]{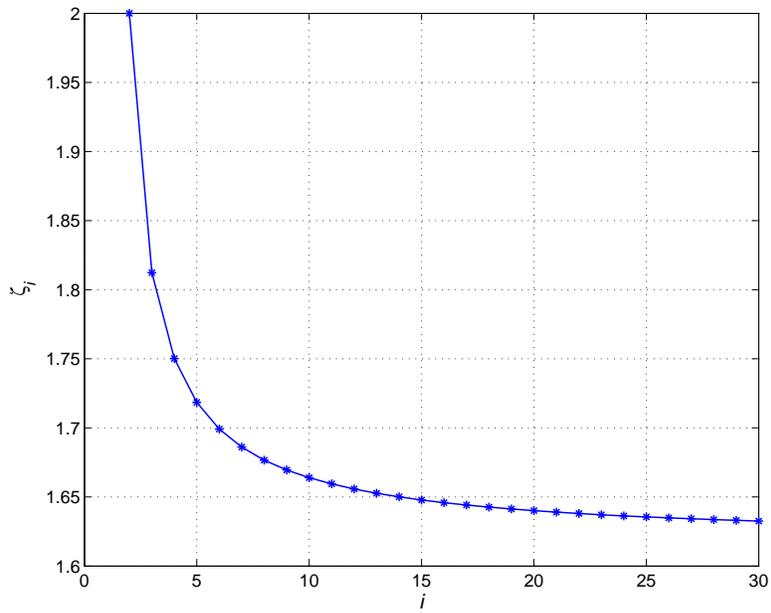}
\caption{The plot of $\zeta_i$ such that $\Psi_{b,i}<\Psi_{b,i+1}$ for $b>\zeta_i$.}
\label{Fig_Thd}
\end{figure}



\end{document}